\documentclass[12pt, onecolumn]{IEEEtran}
\usepackage{epsf}
\usepackage{graphicx}
\usepackage{amsmath} \usepackage{amssymb}

\newtheorem{theorem}{Theorem}[section]
\newtheorem{definition}{Definition}[section]
\newtheorem{lemma}{Lemma}[section]

\newtheorem{remark}{Remark}[section] 
\long\def\symbolfootnote[#1]#2{\begingroup%
\def\thefootnote{\fnsymbol{footnote}}\footnote[#1]{#2}\endgroup}

\def\be{\begin{equation}} \def\ee{\end{equation}}
\def\ba{\begin{array}} \def\ea{\end{array}} \def\bna{\begin{eqnarray}}
\def\ena{\end{eqnarray}}

 \def\bna{\begin{eqnarray}}
\def\ena{\end{eqnarray}} 
\begin{document}
\title{Asymptotic Equipartition Property of Output when Rate is above Capacity}

\author{\authorblockN{Xiugang Wu and Liang-Liang Xie} \\
\authorblockA{\small Department of Electrical and Computer Engineering\\
University of Waterloo, Waterloo, ON, Canada N2L 3G1 \\
Email: x23wu@uwaterloo.ca, llxie@uwaterloo.ca}
\thanks{Part of the work \cite{WuXie09} was presented at CWIT 2009.}}

\maketitle

\begin{abstract}

The output distribution, when rate is above capacity, is investigated. It is shown that there is an asymptotic equipartition property (AEP) of the typical output sequences, independently of the specific codebook used, as long as the codebook is typical according to the standard random codebook generation. This equipartition of the typical output sequences is caused by the mixup of input sequences when there are too many of them, namely, when the rate is above capacity. This discovery sheds some light on the optimal design of the compress-and-forward relay schemes.

\end{abstract}

\section{Introduction}\label{S:Introduction}

A fundamental observation of Shannon's channel coding theorem is that using a randomly generated codebook (i.i.d. generated according to some $p_0(x)$) at a rate below capacity will lead to a distribution pattern of the output sequences, by which, a decoding scheme with arbitrarily low probability of error can be devised.

In this paper, we are interested in the case when the rate is above capacity. We will show that such a pattern that can be used for decoding will disappear when there are too many input sequences, i.e., when the rate is above capacity. Instead, in this case, the output will have an asymptotic equipartition property on the set of typical output sequences (typical with respect to $p_0(y)=\sum_x p_0(x)p(y|x)$). Interestingly, this set is independent of the specific codebook used, as long as the codebook is typical according to the random codebook generation. The reason for this equipartition is that the input sequences are too dense, so that different input sequences can contribute to the same output sequence and get mixed up.

Investigating the optimal compress-and-forward relay scheme has motivated this study of output distribution when rate is above capacity. The optimality of the compress-and-forward schemes is arguably one of the most critical problems in the development of network information theory, where ambiguity always arises when decoding cannot be done correctly. In the classical approach of \cite{covelg79}, the compression scheme at the relay was only based on the distribution used for generating the codebook at the source, instead of the specific codebook generated. While many different codebooks can be generated according to the same distribution, can the knowledge of the specific codebook be helpful? There have been some discussions on this issue (e.g., \cite{XueKumarXie}). Here, in this paper, we show that the observations at the relay are somehow independent of the specific codebook used at the source, and only depend on the distribution by which the codebook is generated.

To further explore the optimality of the compress-and-forward schemes, we compare the rates needed to losslessly compress the relay's observation in two different scenarios: i) the relay uses the knowledge of the source's codebook to do the compression; ii) the relay simply ignores this knowledge. It is shown that the minimum required rates in both scenarios are the same when the rate of the source's codebook is above the capacity of the source-to-relay link.

The remainder of the paper is organized as the following. In Section \ref{S:Pre}, we first introduce some standard definitions of strongly typical sequences, and then give a definition of typical codebooks. Then, we summarize our main results in Section \ref{S:results}, followed by the proof of these results in Section \ref{S:AEP}, \ref{S:WHP} and \ref{S:proofofrate}. Finally, as an application of the results, the optimality of the compress-and-forward schemes is discussed in Section \ref{S:Relay}.

\section{Preliminaries}\label{S:Pre}
Consider a discrete memoryless channel $(\mathcal{X},p(y|x),\mathcal{Y})$ with capacity $C:=\max_{p(x)}I(X;Y)$. Under the random coding framework, a random codebook ${\bf C}$ with respect to $p_0(x)$ with rate $R$ and block length $n$ is defined as
\begin{equation}
{\bf C}:=\left\{ X^n(w)\in \mathcal{X}^n, w=1,\ldots,2^{nR} \right\},
\end{equation}
where each codeword in ${\bf C}$ is an i.i.d. random sequence generated according to a fixed input distribution $p_0(x)$.

It is well known that information can be transmitted with arbitrarily small probability of error for sufficiently large $n$ if $R<C$. In this paper, however, we are interested in the case where the rate is above capacity.

\subsection{Strong Typicality}
We begin with some standard definitions on strong typicality [3, Ch.13].

\begin{definition}
The $\epsilon$-strongly typical set with respect to $p_0(x)$, denoted by $A_{\epsilon,0}^{(n)}(X)$, is the set of sequences $x^n \in \mathcal{X}^n$ satisfying:

1. For all $a \in \mathcal{X}$ with $p_0(a)>0$,
$$\left|\frac{1}{n}N(a|x^n)-p_0(a) \right| < \frac{\epsilon}{|\mathcal{X}|},$$

2. For all $a \in \mathcal{X}$ with $p_0(a)=0$, $N(a|x^n)=0$.

$N(a|x^n)$ is the number of occurrences of $a$ in $x^n$.
\end{definition}

Similarly, we can define the $\epsilon$-strongly typical set with respect to $p_0(y)$ and denote it by $A_{\epsilon,0}^{(n)}(Y)$.

\begin{definition}
The $\epsilon$-strongly typical set with respect to $p_0(x,y)$, denoted by $A_{\epsilon,0}^{(n)}(X,Y)$, is the set of sequences $(x^n,y^n) \in \mathcal{X}^n \times \mathcal{Y}^n$ satisfying:

1. For all $(a,b) \in \mathcal{X} \times \mathcal{Y}$ with $p_0(a,b)>0$,
$$\left|\frac{1}{n}N(a,b|x^n,y^n)-p_0(a,b) \right| < \frac{\epsilon}{|\mathcal{X}||\mathcal{Y}|},$$

2. For all $(a,b) \in \mathcal{X} \times \mathcal{Y}$ with $p_0(a,b)=0$, $$N(a,b|x^n,y^n)=0.$$

$N(a,b|x^n,y^n)$ is the number of occurrences of the pair $(a,b)$ in the pair of sequences $(x^n, y^n)$.
\end{definition}

\begin{definition}
The $\epsilon$-strongly conditionally typical set with the sequence $x^n$ with respect to the conditional distribution $p(y|x)$, denoted by $A_{\epsilon}^{(n)}(Y|x^n)$, is the set of sequences $y^n \in \mathcal{Y}^n$ satisfying:

1. For all $(a,b) \in \mathcal{X} \times \mathcal{Y}$ with $p(b|a)>0$,
\begin{equation}\frac{1}{n}\left|N(a,b|x^n,y^n)- p(b|a)N(a|x^n) \right| \leq \epsilon(1+\frac{1}{|\mathcal{Y}|}),\label{E:CT1}\end{equation}

2. For all $(a,b) \in \mathcal{X} \times \mathcal{Y}$ with $p(b|a)=0$, \begin{equation}N(a,b|x^n,y^n)=0.\label{E:CT2}\end{equation}
\end{definition}

\subsection{Typical Codebooks}
\begin{definition}
For the discrete memoryless channel $(\mathcal{X},p(y|x),\mathcal{Y})$, the channel noise is said to be $\epsilon$-typical if for any given input $x^n$, the output $Y^n$ is $\epsilon$-strongly conditionally typical with $x^n$ with respect to the channel transition function $p(y|x)$, i.e., $Y^n \in A_{\epsilon}^{(n)}(Y|x^{n})$.
\end{definition}

Due to the Law of Large Numbers, the channel noise is ``typical'' with high probability.

Index the sequences in $A_{\epsilon,0}^{(n)}(Y)$ as $y_{\epsilon,0}^{n}(i), i=1,\ldots,M_{\epsilon,0}^{(n)}$, where $M_{\epsilon,0}^{(n)}=|A_{\epsilon,0}^{(n)}(Y)|$. Consider the set $F_{\epsilon,0}(i) \subseteq \mathcal{X}^n$, where each sequence in $F_{\epsilon,0}(i)$ is strongly typical and can reach $y_{\epsilon,0}^{n}(i)$ over a channel with typical noise, i.e.,
\begin{equation*}
F_{\epsilon,0}(i):=\left\{ x^{n} \in A_{\epsilon,0}^{(n)}(X): y_{\epsilon,0}^{n}(i) \in A_{\epsilon}^{(n)}(Y|x^{n})  \right\}.
\end{equation*}

The following notation is useful for defining the typical codebooks.
\begin{align*}
P_{\epsilon,0}(i)&:=\mbox{Pr}(\tilde{X}^{n} \in F_{\epsilon, 0}(i)|\tilde{X}^n \in A_{\epsilon,0}^{(n)}(X)),\\
N_{\epsilon,0}(i|\mathcal{C})&:=\sum_{w=1}^{2^{nR}} \mathbb{I} (x^{n}(w) \in F_{\epsilon, 0}(i)),
\end{align*}
where $\tilde{X}^{n}$ is drawn i.i.d. according to $p_0(x)$ and $\mathbb{I}(A)$ is the indicator function:
\begin{equation*}
\mathbb{I}(A)=
\begin{cases}
1  \text{~~~if $A$ holds,}\\
0  \text{~~~otherwise.}
\end{cases}
\end{equation*}

\begin{definition}\label{D:TCB}
A codebook $$\mathcal{C}=\left\{ x^n(w)\in \mathcal{X}^n,w=1,\ldots,2^{nR}\right\}$$ is said to be $\epsilon$-typical with respect to $p_0(x)$ if \\

\begin{enumerate}
\item $x^{n}(w) \in A_{\epsilon,0}^{(n)}(X), \forall w \in \{1,\dots,2^{nR}\}$,\\
\item $\displaystyle\sup_{i \in \{1,\ldots,M_{\epsilon,0}^{(n)}\}}  \left| \frac{N_{\epsilon,0}(i|\mathcal{C})}{2^{nR}} - P_{\epsilon,0}(i) \right|   \leq \frac{n^{3}R}{2^{nR}}$.
\end{enumerate}

\end{definition}

\section{Main Results}\label{S:results}

The main results of this paper are summarized by the following three theorems. Their proofs are presented in Sections \ref{S:AEP}, \ref{S:WHP} and \ref{S:proofofrate} respectively. The application of these results to the relay channel will be discussed in Section \ref{S:Relay}.

\begin{theorem}\label{T:OdmcU}
Given that an $\epsilon$-typical codebook $\mathcal{C}$ is used and the channel noise is also $\epsilon$-typical, then,\footnote{Same as the notation in \cite{Coverbook}, we say $a_n \doteq b_n$ if $\lim_{n \to \infty} \frac{1}{n} \log \frac{a_n}{b_n}=0$. ``$\stackrel {.}{\geq}$'' and ``$\stackrel {.}{\leq}$'' have similar interpretations.}
\begin{equation*}
\mbox{Pr}(Y^n=y_{\epsilon,0}^{n}(i)|\mathcal{C}) \stackrel{.}{=} 2^{-nH_0(Y)}, \forall i \in \{1,\ldots, M_{\epsilon,0}^{(n)} \},
\end{equation*}
when $R>I_0(X;Y)$, where both $H_0(Y)$ and $I_0(X;Y)$ are calculated according to $p_0(x,y)=p_0(x)p(y|x)$.
\end{theorem}

Throughout this paper, we generate the codebook ${\bf C}$ at random according to $p_0(x)$ and reserve only the $\epsilon$-strongly typical codewords. Then we have Theorem \ref{T:dmcwhp} and \ref{T:rate}.

\begin{theorem}\label{T:dmcwhp}
For any $\epsilon >0$,
\begin{equation}
\mbox{Pr}({\bf C} \text{ is $\epsilon$-typical})\to 1 \ \text{as} \ n \to \infty.
\end{equation}
\end{theorem}

\begin{theorem}\label{T:rate}
Consider the conditional entropy of the channel output given the source's codebook information, namely $H(Y^n|\mathbf{C})$. We have
\begin{equation*}
\lim _{n \to \infty} \frac{1}{n} H(Y^n|\mathbf{C})=
\begin{cases}
H_{0}(Y)  \text{~~~~~~~~~~~~when $R>I_{0}(X;Y)$,}\\
R+H_0(Y|X)  \text{~~~when $R<I_{0}(X;Y)$,}
\end{cases}
\end{equation*}
where $H_0(Y)$, $I_0(X;Y)$ and $H_0(Y|X)$ are all calculated according to $p_0(x,y)=p_0(x)p(y|x)$.

In contrast, without the codebook information, we have
\begin{align*}
\lim_{n\to \infty}\frac{1}{n}H(Y^n)=H_0(Y) \text{ for any } R>0.
\end{align*}
\end{theorem}

\section{AEP of Typical Output Sequences}\label{S:AEP}

Essentially, Theorem \ref{T:OdmcU} states that there exists an asymptotic equipartition property of the typical output sequences, irrespective of the specific codebook used, as long as the codebook is a typical codebook. To prove this theorem, we first introduce two lemmas.

\begin{lemma}\label{L:TransProb}
Let $E_{\epsilon}$ denote the event that the output $Y^n \in A_{\epsilon}^{(n)}(Y|x^n)$ for any given input $x^n$. For any $x^n \in F_{\epsilon, 0}(i)$,
\begin{align*}
&\mbox{Pr}(Y^{n}=y_{\epsilon,0}^{n}(i)|E_{\epsilon}, X^n=x^{n}) \geq 2^{-n(H_0(Y|X)+\epsilon_0)}\\
\text{and~~} &\mbox{Pr}(Y^{n}=y_{\epsilon,0}^{n}(i)|E_{\epsilon}, X^n=x^{n}) \leq 2^{-n(H_0(Y|X)-\epsilon_0)},
\end{align*}
where $H_0(Y|X)$ is calculated according to $p_0(x,y)=p_0(x)p(y|x)$ and $\epsilon_0$ goes to 0 as $\epsilon \to 0$ and $n \to \infty$.
\end{lemma}

\begin{proof}
By the definition of $F_{\epsilon, 0}(i)$, we have for any $x^n$ in $F_{\epsilon, 0}(i)$, $x^n \in A_{\epsilon,0}^{(n)}(X)$ and $y_{\epsilon,0}^{n}(i) \in A_{\epsilon}^{(n)}(Y|x^n)$. Then, it follows from the definition of strong typicality that $(x^n,y_{\epsilon,0}^{n}(i)) \in A_{\epsilon', 0}^{(n)}(X,Y)$, where $\epsilon' \to 0$ as $\epsilon \to 0$. Since strong typicality implies weak typicality, for any $x^n$ in $F_{\epsilon,0}(i)$, we have
\begin{align*}
\left|-\frac{1}{n}\log p(x^n) - H_0(X) \right| < \epsilon'',\\
\left|-\frac{1}{n}\log p(x^n,y_{\epsilon,0}^{n}(i)) - H_0(X,Y) \right| < \epsilon'',
\end{align*}
where $\epsilon'' \to 0$ as $\epsilon \to 0$. Thus,
$$\left|-\frac{1}{n}\log p(y_{\epsilon,0}^{n}(i)|x^n)-H_0(Y|X)  \right| < 2 \epsilon'',$$
and
\begin{equation*}
2^{-n(H_0(Y|X)+2\epsilon'')}      \leq p(y_{\epsilon,0}^{n}(i)|x^n) \leq 2^{-n(H_0(Y|X)-2\epsilon'')}.
\end{equation*}

Therefore, for any $x^n \in F_{\epsilon, 0}(i)$, we have
\begin{align}
\nonumber &\mbox{Pr}(Y^{n}=y_{\epsilon,0}^{n}(i)|E_{\epsilon}, X^n=x^{n})\\
\nonumber =&\frac{\mbox{Pr}(Y^{n}=y_{\epsilon,0}^{n}(i),E_{\epsilon}, X^n=x^{n})}{\mbox{Pr}(E_{\epsilon}, X^n=x^{n})}\\
\nonumber =&\frac{\mbox{Pr}(Y^{n}=y_{\epsilon,0}^{n}(i), X^n=x^{n})}{\mbox{Pr}(E_{\epsilon}| X^n=x^{n})\mbox{Pr}(X^n=x^{n})}\\
\nonumber =&\frac{p(y_{\epsilon,0}^{n}(i)|x^n)}{\mbox{Pr}(E_{\epsilon}|X^n=x^{n})}\\
\nonumber =&(1+o(1))p(y_{\epsilon,0}^{n}(i)|x^n)\\
\nonumber\leq&(1+o(1))2^{-n(H_0(Y|X)-2\epsilon'')} \\
\nonumber=& 2^{-n(H_0(Y|X)-\epsilon_0)},
\end{align}
where $\epsilon_0:=2\epsilon''+\frac{\log (1+o(1))}{n}$ and $\epsilon_0 \to 0$ as $\epsilon \to 0$ and $n \to \infty$. Similarly, for any $x^n \in F_{\epsilon, 0}(i)$, we have
\begin{align*}
&\mbox{Pr}(Y^{n}=y_{\epsilon,0}^{n}(i)|E_{\epsilon}, X^n=x^{n}) \nonumber\\
=&(1+o(1))p(y_{\epsilon,0}^{n}(i)|x^n) \\
\geq &2^{-n(H_0(Y|X)+2\epsilon''-\frac{\log (1+o(1))}{n})}\\
\nonumber\geq& 2^{-n(H_0(Y|X)+\epsilon_0)},
\end{align*}
which finishes the proof of Lemma \ref{L:TransProb}.
\end{proof}

\begin{lemma}\label{L:NinT}
If $\mathcal{C}$ is a typical codebook, then for any $i \in \{1,\ldots,M_{\epsilon,0}^{(n)}\}$,
\begin{align*}
&N_{\epsilon,0}(i| \mathcal{C})\geq 2^{nR} \cdot 2^{-n(I_0(X;Y)+\epsilon'_{0})} -n^{3}R\\
\text{and~~} &N_{\epsilon,0}(i| \mathcal{C})\leq 2^{nR} \cdot 2^{-n(I_0(X;Y)-\epsilon'_0)} + n^{3}R,
\end{align*}
where $I_0(X;Y)$ is calculated according to $p_0(x)p(y|x)$ and $\epsilon'_0$ goes to 0 as $\epsilon \to 0$ and $n \to \infty$.
\end{lemma}

\begin{proof}
To prove Lemma \ref{L:NinT}, we need the following standard result on strong typicality (see Lemma 13.6.2 in \cite{Coverbook}):

Let $\tilde{X}^{n}$ be drawn i.i.d. according to $p_0(x)=\sum_{y}p_0(x,y)$. For $y^{n} \in A_{\epsilon,0}^{(n)}(Y)$,
\begin{align}
&\mbox{Pr}((\tilde{X}^{n}, y^{n}) \in A_{\epsilon,0}^{(n)}(X,Y)) \geq 2^{-n(I_0(X;Y)+\epsilon_1)}\\
\text{and~~} &\mbox{Pr}((\tilde{X}^{n}, y^{n}) \in A_{\epsilon,0}^{(n)}(X,Y)) \leq 2^{-n(I_0(X;Y)-\epsilon_1)},
\end{align}
where $I_0(X;Y)$ is calculated according to $p_0(x,y)$ and $\epsilon_1$ goes to 0 as $\epsilon \to 0$ and $n \to \infty$.

According to the definition of $P_{\epsilon,0}(i)$,
\begin{align*}
&P_{\epsilon,0}(i)=\mbox{Pr}(y_{\epsilon,0}^{n}(i) \in A_{\epsilon}^{(n)}(Y|\tilde{X}^n)|\tilde{X}^n \in A_{\epsilon,0}^{(n)}(X) ),
\end{align*}
where $\tilde{X}^{n}$ is drawn i.i.d. according to $p_0(x)$.

Since $\tilde{X}^n \in A_{\epsilon,0}^{(n)}(X)$ and $y_{\epsilon,0}^{n}(i) \in A_{\epsilon}^{(n)}(Y|\tilde{X}^n)$ imply that $(\tilde{X}^n,y_{\epsilon,0}^{n}(i)) \in A_{\epsilon',0}^{(n)}(X,Y)$, where $\epsilon'$ goes to 0 as $\epsilon \to 0$, we have
\begin{align}
\nonumber P_{\epsilon,0}(i)\leq & \mbox{Pr}((\tilde{X}^{n}, y_{\epsilon,0}^{n}(i)) \in A_{\epsilon',0}^{(n)}(X,Y)|\tilde{X}^{n} \in A_{\epsilon,0}^{(n)}(X) )\\
\nonumber =&\frac{\mbox{Pr}((\tilde{X}^{n}, y_{\epsilon,0}^{n}(i)) \in A_{\epsilon',0}^{(n)}(X,Y),\tilde{X}^{n} \in A_{\epsilon,0}^{(n)}(X) )}{\mbox{Pr}(\tilde{X}^{n} \in A_{\epsilon,0}^{(n)}(X))}\\
\nonumber\leq&(1+o(1))\mbox{Pr}((\tilde{X}^{n}, y_{\epsilon,0}^{n}(i)) \in A_{\epsilon',0}^{(n)}(X,Y))\\
\nonumber \leq&(1+o(1)) 2^{-n(I_0(X;Y)-\epsilon'_1)} \\
\nonumber= & 2^{-n(I_0(X;Y)-\epsilon'_1-\frac{\log (1+o(1))}{n})}\\
= & 2^{-n(I_0(X;Y)-\epsilon'_2)} \label{E:new1}
\end{align}
where $\epsilon'_2:=\epsilon'_1+\frac{\log (1+o(1))}{n}$ and $\epsilon'_2 \to 0$ as $\epsilon \to 0$ and $n \to \infty$.

Furthermore, by the standard definitions of strong typicality, it follows that  $(x^n,y_{\epsilon,0}^{n}(i)) \in A_{\epsilon,0}^{(n)}(X,Y)$ implies $x^n \in A_{\epsilon,0}^{(n)}(X)$. Now, we show $(x^n,y_{\epsilon,0}^{n}(i)) \in A_{\epsilon,0}^{(n)}(X,Y)$ also implies $y_{\epsilon,0}^{n}(i) \in A_{\epsilon}^{(n)}(Y|x^n)$. Suppose $(x^n,y_{\epsilon,0}^{n}(i)) \in A_{\epsilon,0}^{(n)}(X,Y)$. Then, we have
\begin{enumerate}
  \item For all $(a,b)\in \mathcal{X}   \times \mathcal{Y}$ with $p(b|a)=0$, $p_0(a,b)=0$ and $N(a,b|x^n,y_{\epsilon,0}^{n}(i))=0$.
  \item For all $(a,b)\in \mathcal{X}   \times \mathcal{Y}$ with $p(b|a)>0$ and $p_0(a)=0$, $p_0(a,b)=0$ and $N(a,b|x^n,y_{\epsilon,0}^{n}(i))=0$, as well as $N(a|x^n)=0$.
  \item For all $(a,b)\in \mathcal{X}   \times \mathcal{Y}$ with $p(b|a)>0$ and $p_0(a)>0$, $p_0(a,b)>0$ and
             $$\left|\frac{1}{n}N(a|x^n)-p_0(a) \right| < \frac{\epsilon}{|\mathcal{X}|},$$
             $$\left|\frac{1}{n}N(a,b|x^n,y_{\epsilon,0}^{n}(i))-p_0(a,b) \right| < \frac{\epsilon}{|\mathcal{X}||\mathcal{Y}|}.$$
  Thus,
  \begin{align*}
  &\left|\frac{1}{n}N(a,b|x^n,y_{\epsilon,0}^{n}(i))-\frac{1}{n}N(a|x^n)p(b|a) \right|\\
  <&p_0(a,b)+\frac{\epsilon}{|\mathcal{X}||\mathcal{Y}|}-p(b|a)(p_0(a)-\frac{\epsilon}{|\mathcal{X}|})\\
  =&\frac{\epsilon}{|\mathcal{X}||\mathcal{Y}|}+p(b|a)\frac{\epsilon}{|\mathcal{X}|}\\
  \leq&\frac{\epsilon}{|\mathcal{Y}|}+ \epsilon\\
  =& \epsilon(1+\frac{1}{|\mathcal{Y}|}).
  \end{align*}
\end{enumerate}
Therefore, $(x^n,y_{\epsilon,0}^{n}(i)) \in A_{\epsilon,0}^{(n)}(X,Y)$ implies that $y_{\epsilon,0}^{n}(i) \in A_{\epsilon}^{(n)}(Y|x^n)$, as well as $x^n \in A_{\epsilon,0}^{(n)}(X)$.

Then, we have
\begin{align}
\nonumber P_{\epsilon,0}(i)=&\mbox{Pr}(y_{\epsilon,0}^{n}(i) \in A_{\epsilon}^{(n)}(Y|\tilde{X}^n)|\tilde{X}^n \in A_{\epsilon,0}^{(n)}(X) )\\
\nonumber\geq&\mbox{Pr}((\tilde{X}^{n}, y_{\epsilon,0}^{n}(i)) \in A_{\epsilon,0}^{(n)}(X,Y)|\tilde{X}^{n} \in A_{\epsilon,0}^{(n)}(X) )\\
\nonumber =&\frac{\mbox{Pr}((\tilde{X}^{n}, y_{\epsilon,0}^{n}(i)) \in A_{\epsilon,0}^{(n)}(X,Y),\tilde{X}^{n} \in A_{\epsilon,0}^{(n)}(X) )}{\mbox{Pr}(\tilde{X}^{n} \in A_{\epsilon,0}^{(n)}(X))}\\
\nonumber=&(1+o(1))\mbox{Pr}((\tilde{X}^{n}, y_{\epsilon,0}^{n}(i)) \in A_{\epsilon,0}^{(n)}(X,Y))\\
\nonumber\geq&(1+o(1)) 2^{-n(I_0(X;Y)+\epsilon_1)}\\
\nonumber=& 2^{-n(I_0(X;Y)+\epsilon_1-\frac{\log (1+o(1))}{n})}\\
= & 2^{-n(I_0(X;Y)+\epsilon_{2})}\label{E:new2}
\end{align}
where $\epsilon_{2}:=\epsilon_1-\frac{\log (1+o(1))}{n}$ and $\epsilon_2 \to 0$ as $\epsilon \to 0$ and $n \to \infty$.

Let $\epsilon'_0=\max\{\epsilon_{2},\epsilon'_{2}\}$. Combining (\ref{E:new1}) and (\ref{E:new2}), we have
\begin{align}
2^{-n(I_0(X;Y)+\epsilon'_0)}     \leq P_{\epsilon,0}(i) \leq 2^{-n(I_0(X;Y)-\epsilon'_0)}.\label{E:combinebound}
\end{align}

Therefore, if $\mathcal{C}$ is a typical codebook, by the definition of the typical codebooks and (\ref{E:combinebound}), for any $i \in \{1,\ldots,M_{\epsilon,0}^{(n)}\}$,
\begin{align*}
&N_{\epsilon,0}(i| \mathcal{C})\geq 2^{nR} \cdot 2^{-n(I_0(X;Y)+\epsilon'_{0})} -n^{3}R\\
\text{and~~} &N_{\epsilon,0}(i| \mathcal{C})\leq 2^{nR} \cdot 2^{-n(I_0(X;Y)-\epsilon'_0)} + n^{3}R,
\end{align*}
where $I_0(X;Y)$ is calculated according to $p_0(x)p(y|x)$ and $\epsilon'_0$ goes to 0 as $\epsilon \to 0$ and $n \to \infty$.
\end{proof}

\begin{proof}[Proof of Theorem \ref{T:OdmcU}]
Let $E_{\epsilon}$ denote the event $Y^n \in A_{\epsilon}^{(n)}(Y|x^n)$ for any given input $x^n$. Consider $\mbox{Pr}(Y^{n}=y_{\epsilon,0}^{n}(i)|E_{\epsilon}, \mathcal{C} \text{ is typical})$ for any $i \in \{1,\ldots,M_{\epsilon,0}^{(n)} \}$. We lower bound this probability as follows:
\begin{align}
\nonumber & \mbox{Pr}(Y^{n}=y_{\epsilon,0}^{n}(i)|E_{\epsilon}, \mathcal{C}\text{ is typical})\\
\nonumber  \stackrel {}{=} & \sum_{w=1}^{2^{nR}}\mbox{Pr}(Y^{n}=y_{\epsilon,0}^{n}(i)|E_{\epsilon}, \mathcal{C}\text{ is typical},X^n=x^{n}(w))\\
 & \cdot \mbox{Pr}(X^n=x^{n}(w)|E_{\epsilon}, \mathcal{C}\text{ is typical})\label{(a)}\\
 \stackrel {}{=} & \frac{1}{2^{nR}}\sum_{w=1}^{2^{nR}}\mbox{Pr}(Y^{n}=y_{\epsilon,0}^{n}(i)|E_{\epsilon}, \mathcal{C}\text{ is typical},X^n=x^{n}(w))\label{(b)}\\
 \stackrel {}{=} &\frac{1}{2^{nR}}\sum_{x^n(w)\in F_{\epsilon,0}(i)} \mbox{Pr}(Y^{n}=y_{\epsilon,0}^{n}(i)|E_{\epsilon}, \mathcal{C}\text{ is typical},X^n=x^{n}(w)) \label{E:ctn}\\
 \stackrel {}{\geq} &\frac{1}{2^{nR}}  N_{\epsilon,0}(i| \mathcal{C}) \cdot 2^{-n(H_0(Y|X)+\epsilon_0)} \label{(d)} \\
\stackrel {}{\geq} &\frac{1}{2^{nR}} (2^{nR} \cdot 2^{-n(I_0(X;Y)+\epsilon'_0)} -n^{3}R)  \cdot 2^{-n(H_0(Y|X)+\epsilon_0)} \label{(e)} \\
\nonumber \stackrel {}{=} &2^{-n(H_0(Y)+\epsilon_0+\epsilon'_0)}\cdot \left[1 - \frac{n^{3}R}{2^{nR}} \cdot 2^{n(I_0(X;Y)+\epsilon'_0)} \right].
\end{align}

(\ref{(a)}) follows from the Law of Total Probability and accumulates the contributions from all the codewords in the codebook to the probability for $y_{\epsilon,0}^{n}(i)$ to be channel output.

(\ref{(b)}) follows from the uniform distribution of message index $W$.

(\ref{E:ctn}) follows from the the condition $E_{\epsilon}$ and the fact that $\mathcal{C}$ contains only strongly typical codewords.

(\ref{(d)}) follows from Lemma \ref{L:TransProb}.

(\ref{(e)}) follows from Lemma \ref{L:NinT}.

Let $\epsilon \to 0$ as $n \to \infty$. Then for any $i \in \{1,\ldots,M_{\epsilon,0}^{(n)}\}$,
\begin{align}
&\mbox{Pr}(Y^{n}=y_{\epsilon}^{n}(i)|E_{\epsilon},  \mathcal{C} \text{ is typical})\stackrel {.}{\geq}2^{-nH_0(Y)},\label{E:Lower}
\end{align}
when $R>I_0(X;Y)$.

Similarly, following (\ref{E:ctn}), by Lemmas \ref{L:TransProb} and \ref{L:NinT}, we have
\begin{align*}
& \mbox{Pr}(Y^{n}=y_{\epsilon,0}^{n}(i)|E_{\epsilon}, \mathcal{C}\text{ is typical})\\
\stackrel {}{\leq} &\frac{1}{2^{nR}}  N_{\epsilon,0}(i| \mathcal{C}) \cdot 2^{-n(H_0(Y|X)-\epsilon_0)}  \\
\stackrel {}{\leq} &\frac{1}{2^{nR}} (2^{nR} \cdot 2^{-n(I_0(X;Y)-\epsilon'_0)} +n^{3}R)  \cdot 2^{-n(H_0(Y|X)-\epsilon_0)}  \\
\stackrel {}{=} &2^{-n(H_0(Y)-\epsilon_0-\epsilon'_0)}\cdot \left[1 + \frac{n^{3}R}{2^{nR}} \cdot 2^{n(I_0(X;Y)-\epsilon'_0)} \right].
\end{align*}
Therefore, for any $i \in \{1,\ldots,M_{\epsilon,0}^{(n)}\}$,
\begin{align}
&\mbox{Pr}(Y^{n}=y_{\epsilon,0}^{n}(i)|E_{\epsilon},  \mathcal{C} \text{ is typical})\stackrel {.}{\leq}2^{-nH_0(Y)},\label{E:Upper}
\end{align}
when $R>I_0(X;Y)$.
Combining (\ref{E:Lower}) and (\ref{E:Upper}), we establish Theorem \ref{T:OdmcU}.
\end{proof}

\section{The Probability that A Typical Codebook Appears}\label{S:WHP}

In this section, we will show that with high probability, a typical codebook will be generated by the random codebook generation. We begin with some relevant definitions and the Vapnik-Chervonenkis Theorem \cite{vctheorem1}, \cite{vctheorem2}:

A Range Space is a pair $(X,\mathcal{F})$, where $X$ is a set and $\mathcal{F}$ is a family of subsets of $X$. For any $A \subseteq X$, we define $P_{\mathcal{F}}(A)$, the projection of $\mathcal{F}$ on $A$, as $\{F \cap A: F \in \mathcal{F}\}$. We say that $A$ is \emph{shattered} by $\mathcal{F}$ if $P_{\mathcal{F}}(A) = 2^{A}$, i.e., if the projection of $\mathcal{F}$ on $A$ includes all possible subsets of $A$. The VC-dimension of $\mathcal{F}$, denoted by VC-d($\mathcal{F}$) is the cardinality of the largest set $A$ that $\mathcal{F}$ shatters. If arbitrarily large finite sets are shattered, the VC dimension of $\mathcal{F}$ is infinite.

\emph{The Vapnik-Chervonenkis Theorem:}
If $\mathcal{F}$ is a set of finite VC-dimension and $\{Y_{j}\}$ is a sequence of $n$ i.i.d. random variables with common probability distribution $P$, then for every $\epsilon$, $ \delta> 0$
\begin{equation} \label{E:vctheorem1}
\mbox{Pr} \left\{ \sup_{F \in \mathcal{F}} \left| \frac{1}{n} \sum_{j=1}^{n} \mathbb{I}(Y_{j}\in F) - P(F) \right| \leq \epsilon  \right\} >1-\delta
\end{equation}
whenever
\begin{equation}\label{E:vctheorem2}
n>\max\left\{\frac{\text{8VC-d}(\mathcal{F})}{\epsilon}\log_{2}\frac{16e}{\epsilon}  , \frac{4}{\epsilon}\log_{2}\frac{2}{\delta}       \right  \}.
\end{equation}

Let $\mathcal{F}_{\epsilon,0}=\{F_{\epsilon,0}(i), i=1,\ldots,M_{\epsilon,0}^{(n)} \}$. To show Theorem \ref{T:dmcwhp}, a finite VC dimension of $\mathcal{F}_{\epsilon,0}$ is desired in order to employ the Vapnik-Chervonenkis Theorem. For this reason, we introduce Lemma \ref{L:dmcvcdimension}.

\begin{lemma}\label{L:dmcvcdimension}
For a fixed block length $n$, VC-d$(\mathcal{F}_{\epsilon,0}) \leq n(H_0(Y)+\epsilon ')$, where $\epsilon ' \to 0$ as $\epsilon \to 0$.
\end{lemma}
\begin{proof}\label{P:vcdimension}
By the Asymptotic Equipartition Property, $|\mathcal{F}_{\epsilon,0}|=M_{\epsilon,0}^{(n)}\leq  2^{n(H_0(Y)+\epsilon')}$, where $\epsilon' \to 0$ as $\epsilon \to 0$. Thus, for any $A \subseteq \mathcal{X}^{n}$,
$$|\{F_{\epsilon,0}(i)\cap A: F_{\epsilon,0}(i) \in \mathcal{F}_{\epsilon,0}\}| \leq 2^{n(H_0(Y)+\epsilon ')},$$
and hence VC-d$(\mathcal{F}_{\epsilon,0}) \leq n(H_0(Y)+\epsilon ')$.
\end{proof}

\begin{proof}[Proof of Theorem \ref{T:dmcwhp}]\label{P:dmcwhp}
Since we reserve only the $\epsilon$-strongly typical codewords when generating the codebook, for any random codebook, the first condition in Definition \ref{D:TCB} is obviously satisfied. Below, we focus on showing that a random codebook satisfies the second condition in Definition \ref{D:TCB} with high probability.

For the given $p_0(x)$, consider all the codewords in a random codebook, $X^{n}(w)$, $w =1,\ldots,2^{nR}$. They are generated with the common distribution $p(x^n)=\mbox{Pr}(\tilde{X}^{n} = x^n|\tilde{X}^n \in A_{\epsilon,0}^{(n)}(X))$, where $\tilde{X}^{n}$ is drawn i.i.d. according to $p_0(x)$. Since VC-d$(\mathcal{F}_{\epsilon,0})$ is finite for a fixed $n$, we employ the Vapnik-Chervonenkis Theorem under the range space $(\mathcal{X}^{n},\mathcal{F}_{\epsilon,0})$. To satisfy (\ref{E:vctheorem2}), let both $\epsilon$ and $\delta$ in (\ref{E:vctheorem1}) be $\frac{\Delta_{\epsilon}nR}{2^{nR}}$, where $\Delta_{\epsilon}:=\max \{ 8\text{VC-d}(\mathcal{F}_{\epsilon,0}), 16e  \}.$ Then the Vapnik-Chervonenkis Theorem states that
\begin{align}
\nonumber &\mbox{Pr}\left\{\sup_{F_{\epsilon,0}(i) \in \mathcal{F}_{\epsilon,0}}  \left| \frac{N_{\epsilon,0}(i|{\bf C})}{2^{nR}} - P_{\epsilon,0}(i) \right|   \leq \frac{\Delta_{\epsilon}nR}{2^{nR}}     \right\}\\
\nonumber \geq &1-\frac{\Delta_{\epsilon}nR}{2^{nR}}\\
\to & 1 \ \text{as} \ n \to \infty , \label{E:O11}
\end{align}
where $N_{\epsilon,0}(i|{\bf C})=\sum_{w=1}^{2^{nR}}\mathbb{I}(X^{n}(w)\in F_{\epsilon,0}(i))$. Since $\frac{n^3 R}{2^{nR}}  \geq \frac{\Delta_{\epsilon}nR}{2^{nR}} $ for sufficiently large $n$, (\ref{E:O11}) concludes the proof of Theorem \ref{T:dmcwhp}.
\end{proof}

\section{Proof of Theorem \ref{T:rate}}\label{S:proofofrate}

Before proceeding to the proof of Theorem \ref{T:rate}, we first introduce Lemma \ref{L:JAEP}, which will facilitate the later discussions. The proof of Lemma \ref{L:JAEP} is given in Appendix \ref{A:A}.

\begin{lemma}\label{L:JAEP}
For the channel $(\mathcal{X},p(y|x),\mathcal{Y})$, generate the codebook at random according to $p_0(x)$ and reserve only the $\epsilon$-strongly typical codewords. The channel input and output $X^n$ and $Y^n$ satisfy that

\begin{enumerate}
   \item $\mbox{Pr}((X^n,Y^n) \in A_{\epsilon,0}^{(n)}(X,Y))\to 1 \ \text{as} \ n \to \infty$, for any $\epsilon >0$;
   \item $\lim _{n \to \infty} \frac{1}{n} H(X^n)=H_0(X)$, $\lim _{n \to \infty} \frac{1}{n} H(Y^n)=H_0(Y)$, and $\lim _{n \to \infty} \frac{1}{n} H(X^n,Y^n)=H_0(X,Y)$.
\end{enumerate}

\end{lemma}

\begin{remark}
Since we reserve only the $\epsilon$-typical codewords when generating the codebook, generally, the channel input $X^n$ is no longer an i.i.d. random process. However, Lemma \ref{L:JAEP} essentially states that the random process $(X^n,Y^n)$ still satisfies the joint asymptotic equipartition property and furthermore, the entropy rates of the random processes $X^n$, $Y^n$ and $(X^n,Y^n)$ can still be simply expressed in the single letter form respectively. This observation will facilitate our later discussions.
\end{remark}

\begin{proof}[Proof of Theorem \ref{T:rate}]
We prove Theorem \ref{T:rate} by characterizing $\lim _{n \to \infty} \frac{1}{n} H(Y^n|\mathbf{C})$ in two different cases: when $R>I_0(X;Y)$ and when $R<I_0(X;Y)$, respectively.

\subsection{When $R>I_0(X;Y)$}

Define an indicator random variable $E$ as
\begin{align*}
E:=\mathbb{I}(E_{\epsilon}),
\end{align*}
where $E_{\epsilon}$ denotes the event $Y^n \in A_{\epsilon}^{(n)}(Y|x^n)$ for any given input $x^n$.

When $R>I_{0}(X;Y)$, we have
\begin{align}
\nonumber &H(Y^n|\mathbf{C})\\
\geq & H(Y^n|E,\mathbf{C})\label{(f)}\\
\nonumber =&\mbox{Pr}(E=1)H(Y^n|E=1,\mathbf{C}) + \mbox{Pr}(E=0)H(Y^n|E=0,\mathbf{C})\\
\nonumber \geq & \mbox{Pr}(E=1) \cdot H(Y^n|E=1,\mathbf{C})\\
\stackrel {}{=} & (1-o(1)) \cdot H(Y^n|E=1,\mathbf{C})\label{(g)}\\
\nonumber=& (1-o(1)) \cdot \sum_{\mathcal{C}}p(\mathcal{C}) \cdot    H(Y^n|E=1,\mathbf{C}=\mathcal{C})\\
\nonumber\geq & (1-o(1)) \cdot  \sum_{\mathcal{C} \text{ is typical}} p(\mathcal{C}) \cdot    H(Y^n|E=1,\mathbf{C}=\mathcal{C} )\\
\nonumber\stackrel { }{=}&(1-o(1)) \cdot \sum_{\mathcal{C}\text{ is typical}} p(\mathcal{C}) \cdot \left(\sum_{y^n}p(y^n|E_{\epsilon},\mathcal{C}) \log \frac{1}{p(y^n|E_{\epsilon},\mathcal{C})}        \right)\\
\nonumber \geq & (1-o(1))\cdot\sum_{\mathcal{C}\text{ is typical}} p(\mathcal{C}) \cdot \left(\sum_{y^n \in A_{\epsilon,0}^{(n)}(Y)}p(y^n|E_{\epsilon},\mathcal{C}) \log \frac{1}{p(y^n|E_{\epsilon},\mathcal{C})}   \right)\\
\stackrel {}{\geq} & (1-o(1)) \cdot \sum_{\mathcal{C}\text{ is typical}} p(\mathcal{C}) \cdot \left(\sum_{y^n \in A_{\epsilon,0}^{(n)}(Y)}p(y^n|E_{\epsilon},\mathcal{C}) \log 2^{n[H_0(Y)-\epsilon^*]}\right) \label{(h)}  \\
\nonumber =& n[H_0(Y)-\epsilon^*]\cdot(1-o(1)) \cdot \sum_{\mathcal{C}\text{ is typical}} p(\mathcal{C}) \cdot \left(\sum_{y^n \in A_{\epsilon,0}^{(n)}(Y)}p(y^n|E_{\epsilon},\mathcal{C})  \right)   \\
\nonumber =& n[H_0(Y)-\epsilon^*]\cdot(1-o(1)) \cdot \sum_{\mathcal{C}\text{ is typical}} p(\mathcal{C}) \cdot \mbox{Pr}(Y^n\in A_{\epsilon,0}^{(n)}(Y)| E_{\epsilon},\mathcal{C}) \\
\nonumber =& n[H_0(Y)-\epsilon^*]\cdot(1-o(1)) \cdot \sum_{\mathcal{C}\text{ is typical}} p(\mathcal{C}|E_{\epsilon}) \cdot \mbox{Pr}(Y^n\in A_{\epsilon,0}^{(n)}(Y)| E_{\epsilon},\mathcal{C}) \\
\nonumber \stackrel {}{=} & n[H_0(Y)-\epsilon^*]\cdot(1-o(1)) \cdot \mbox{Pr}(Y^n\in A_{\epsilon,0}^{(n)}(Y),\mathbf{C}\text{ is typical}| E_{\epsilon}) \\
\stackrel {}{=}& n[H_0(Y)-\epsilon^*]\cdot(1-o(1)) \cdot  (1-o(1))\label{(i)}  \\
\nonumber=& n[H_0(Y)-\epsilon^*]\cdot(1-o(1))
\end{align}

(\ref{(f)}) follows from the fact that conditioning reduces entropy.

(\ref{(g)}) follows from the fact that $\mbox{Pr}(E_{\epsilon}) \to 1$ as $n \to \infty$, for any $\epsilon>0$.

(\ref{(h)}) follows from Theorem \ref{T:OdmcU}, which upper bounds $p(y^n|E_{\epsilon}, \mathcal{C})$ by $2^{-n[H_0(Y)-\epsilon^*]}$ for any $y^n\in A_{\epsilon,0}^{(n)}(Y)$ and typical $\mathcal{C}$, where $\epsilon^{*} \to 0$ as $n \to \infty$.

(\ref{(i)}) follows from the fact that $$\mbox{Pr}(Y^n\in A_{\epsilon,0}^{(n)}(Y),\mathbf{C}\text{ is typical}| E_{\epsilon}) \to 1 \text{ as } n \to \infty.$$ This can be seen from the following.
\begin{align}
\nonumber &\mbox{Pr}(Y^n\in A_{\epsilon,0}^{(n)}(Y),\mathbf{C}\text{ is typical}| E_{\epsilon}) \\
\nonumber =&\frac{\mbox{Pr}(Y^n\in A_{\epsilon,0}^{(n)}(Y), E_{\epsilon},\mathbf{C}\text{ is typical})}{\mbox{Pr}(E_{\epsilon})}\\
\geq & \frac{\mbox{Pr}((X^n,Y^n)\in A_{\epsilon,0}^{(n)}(X,Y), E_{\epsilon},\mathbf{C}\text{ is typical})}{\mbox{Pr}(E_{\epsilon})}\label{L:NDgoto1}.
\end{align}
Since  $\mbox{Pr}(E_{\epsilon})$, $\mbox{Pr}(\mathbf{C}\text{ is typical})$ and $\mbox{Pr}((X^n,Y^n)\in A_{\epsilon,0}^{(n)}(X,Y))$ all go to 1, obviously  both the numerator and denominator of (\ref{L:NDgoto1}) go to 1 as $n \to \infty$. Thus, $$\mbox{Pr}(Y^n\in A_{\epsilon,0}^{(n)}(Y),\mathbf{C}\text{ is typical}| E_{\epsilon}) \to 1 \text{ as } n \to \infty.$$

Therefore, when $R>I_0(X;Y)$,
\begin{align}
\nonumber &\liminf _{n \to \infty} \frac{1}{n} H(Y^n|\mathbf{C})\\
\nonumber \geq &\liminf _{n \to \infty} \frac{1}{n}\left( n[H_0(Y)-\epsilon^*]\cdot(1-o(1))\right)\\
\nonumber =&\liminf _{n \to \infty} [H_0(Y)-\epsilon^*]\cdot(1-o(1))\\
=& H_0(Y).\label{E:sandwich 1}
\end{align}

Furthermore,
\begin{align}
 \limsup _{n \to \infty} \frac{1}{n} H(Y^n|\mathbf{C}) \leq & \limsup _{n \to \infty} \frac{1}{n} H(Y^n)=H_0(Y), \label{E:sandwich 2}
\end{align}
where the last equality follows from Lemma \ref{L:JAEP}.

Combining (\ref{E:sandwich 1}) and (\ref{E:sandwich 2}), we have that when $R>I_0(X;Y)$,
\begin{align*}
\lim _{n \to \infty} \frac{1}{n} H(Y^n|\mathbf{C})=H_0(Y).
\end{align*}

\subsection{When $R<I_0(X;Y)$}

To find $\lim _{n \to \infty} \frac{1}{n} H(Y^n|\mathbf{C})$ when $R<I_0(X;Y)$, we first introduce two lemmas. The proofs of these two lemmas are given at the end of this section.

\begin{lemma}\label{L:rlessc1}
When $R<I_0(X;Y)$,
\begin{align*}\frac{1}{n}H(X^n|\mathbf{C},Y^n)\to 0, \text{ as } n \to \infty.\end{align*}
\end{lemma}

\begin{lemma}\label{L:rlessc2}
\begin{align*}
\lim _{n \to \infty} \frac{1}{n} H(X^n|\mathbf{C})= R.
\end{align*}
\end{lemma}

Now, expanding $H(X^n,Y^n|\mathbf{C})$ in two different ways, we have
\begin{align*}
H(X^n,Y^n|\mathbf{C})=&H(X^n|\mathbf{C})+H(Y^n|X^n,\mathbf{C})\\
=&H(Y^n|\mathbf{C})+H(X^n|\mathbf{C},Y^n),
\end{align*}
and thus
\begin{align*}
H(Y^n|\mathbf{C})=H(X^n|\mathbf{C})+H(Y^n|X^n,\mathbf{C})-H(X^n|\mathbf{C},Y^n).
\end{align*}
Therefore, when $R<I_0(X;Y)$,
\begin{align}
\nonumber \lim _{n \to \infty} \frac{1}{n} H(Y^n|\mathbf{C})=&\lim _{n \to \infty} \frac{1}{n} H(X^n|\mathbf{C})+\lim _{n \to \infty} \nonumber \nonumber \frac{1}{n}H(Y^n|X^n,\mathbf{C})-\lim _{n \to \infty} \frac{1}{n}H(X^n|\mathbf{C},Y^n) \\
=&R+\lim _{n \to \infty} \frac{1}{n}H(Y^n|X^n,\mathbf{C})\label{(j)}\\
=&R+\lim _{n \to \infty} \frac{1}{n}H(Y^n|X^n)\label{(k)}\\
\nonumber =&R+\lim _{n \to \infty} \frac{1}{n}[H(X^n,Y^n)-H(X^n)]\\
=&R+H_0(X,Y)-H_0(X)\label{(l)}\\
\nonumber =&R+H_0(Y|X),
\end{align}
where (\ref{(j)}) follows from Lemma \ref{L:rlessc1} and \ref{L:rlessc2}, (\ref{(k)}) follows from the fact that $\mathbf{C}\rightarrow X^n\rightarrow Y^n$ forms a Markov Chain, and (\ref{(l)}) follows from Lemma \ref{L:JAEP}. This completes the proof of Theorem \ref{T:rate}.
\end{proof}

\begin{proof}[Proof of Lemma \ref{L:rlessc1}]
To prove Lemma \ref{L:rlessc1}, we begin with Fano's Inequality (see Theorem 2.11.1 in \cite{Coverbook}):

Let $P_e=\mbox{Pr}(g(Y)\neq X)$, where $g$ is any function of $Y$. Then
\begin{equation}
1+P_e \log|\mathcal{X}|\geq H(X|Y).
\end{equation}

For the channel $(\mathcal{X},p(y|x),\mathcal{Y})$ with a codebook $\mathcal{C}$, we estimate the message index $W$ from $Y^n$. Let the estimate be
$\hat{W}=g(Y^n)$ and $P_e^{(n)}(\mathcal{C})=\mbox{Pr}(W \neq g(Y^n)|\mathcal{C})$. Then, applying Fano's Inequality, we have
\begin{align*}
H(W|Y^n,\mathcal{C}) \leq 1+ P_e^{(n)}(\mathcal{C})\log 2^{nR}=1+ P_e^{(n)}(\mathcal{C})nR.
\end{align*}

Since given $\mathcal{C}$, $X^n$ is a function of $W$, say $X^n=X^n(W)$, we have
\begin{align*}
H(X^n|Y^n,\mathcal{C})\leq H(W|Y^n,\mathcal{C}) \leq 1+ P_e^{(n)}(\mathcal{C})nR.
\end{align*}
Then,
\begin{align*}
H(X^n|Y^n,\mathbf{C})=\sum_{\mathcal{C}}p(\mathcal{C}) H(W|Y^n,\mathcal{C}) \leq\sum_{\mathcal{C}}p(\mathcal{C})( 1+ P_e^{(n)}(\mathcal{C})nR).
\end{align*}

Recall the channel coding theorem, which states that if we randomly generate the codebook according to $p_0(x)$, then when $R<I_0(X;Y)$,
\begin{align}
\sum_{\mathcal{C}}p(\mathcal{C})P_e^{(n)}(\mathcal{C})\to 0. \label{E:errorprobability}
\end{align}

Therefore, when $R<I_0(X;Y)$,
\begin{align*}
\limsup_{n\to \infty}\frac{1}{n}H(X^n|Y^n,\mathbf{C}) \leq & \limsup_{n\to \infty}\frac{1}{n} \sum_{\mathcal{C}}p(\mathcal{C})[ 1+ P_e^{(n)}(\mathcal{C})nR]\\
=& \limsup_{n\to \infty}\frac{1}{n} [1+ nR\sum_{\mathcal{C}}p(\mathcal{C})P_e^{(n)}(\mathcal{C})     ]\\
=&\limsup_{n\to \infty}\frac{1}{n}+\limsup_{n\to \infty}R\sum_{\mathcal{C}}p(\mathcal{C})P_e^{(n)}(\mathcal{C}) \\
=&0.
\end{align*}
Furthermore, it is obvious that $\frac{1}{n}H(X^n|Y^n,\mathbf{C})\geq 0$ and hence
\begin{align*}\frac{1}{n}H(X^n|\mathbf{C},Y^n)\to 0, \text{ as } n \to \infty,\end{align*}
when $R<I_0(X;Y)$.
\end{proof}

\begin{proof}[Proof of Lemma \ref{L:rlessc2}]
Given any $\mathcal{C}$, $X^n$ is a function of $W$. Thus, $H(X^n|\mathcal{C}) \leq H(W|\mathcal{C})=nR,$
and
\begin{align}
\frac{1}{n}H(X^n|\mathbf{C})=&\frac{1}{n}\sum_{\mathcal{C}}p(\mathcal{C})H(X^n|\mathcal{C})\leq R.\label{E:cbound1}
\end{align}

Therefore, to show Lemma \ref{L:rlessc2}, it suffices to show that $\lim_{n \to \infty} \frac{1}{n}H(X^n|\mathbf{C})\geq R$. For this purpose, we first define a class of codebooks as regular codebooks and focus on characterizing $H(X^n|\mathcal{C})$ for a regular codebook $\mathcal{C}$. Then, we show that a regular codebook appears with high probability when we randomly generate the codebook, and conclude that $\lim_{n \to \infty} \frac{1}{n}H(X^n|\mathbf{C})\geq R$.

We say a codebook $\mathcal{C}$ is regular if
\begin{align*}\sup_{x^n \in A_{\epsilon,0}^{(n)}(X) }  \left| \frac{N(x^n|\mathcal{C})}{2^{nR}} - p(x^n) \right|   \leq \frac{n^{3}R}{2^{nR}},\end{align*}
where $N(x^n|\mathcal{C})$ is the number of occurrences of $x^n$ in $\mathcal{C}$, defined by
\begin{align*}
N(x^n|\mathcal{C})=\sum _{w=1}^{2^{nR}}\mathbb{I}(x^n(w)=x^n),
\end{align*}
and $p(x^n)=\mbox{Pr}(\tilde{X}^{n} = x^n|\tilde{X}^n \in A_{\epsilon,0}^{(n)}(X))$ where $\tilde{X}^{n}$ is drawn i.i.d. according to $p_0(x)$.

Given a regular $\mathcal{C}$, for any $x^n \in A_{\epsilon,0}^{(n)}(X) $, we have
\begin{align}
\nonumber N(x^n|\mathcal{C}) \leq& 2^{nR}p(x^n) + n^{3}R\\
\nonumber =& 2^{nR}\mbox{Pr}(\tilde{X}^{n} = x^n|\tilde{X}^n \in A_{\epsilon,0}^{(n)}(X)) + n^{3}R\\
\leq& 2^{nR}(1+o(1)) 2^{-n(H_0(X)-\epsilon')}+ n^{3}R\label{epsilon}\\
=&n^{3}R+o(1), \label{wlog}
\end{align}
where the $\epsilon'$ in (\ref{epsilon}) goes to 0 as $\epsilon \to 0$ and (\ref{wlog}) follows from the general assumption that $R < H_0(X)$.
Note that the message index $W$ is uniformly distributed, we have for a given $\mathcal{C}$ and any $x^n \in A_{\epsilon,0}^{(n)}(X) $,
\begin{align*}
p(x^n|\mathcal{C})=&\frac{\sum _{w=1}^{2^{nR}}\mathbb{I}(x^n(w)=x^n)}{2^{nR}}\\
=&\frac{N(x^n|\mathcal{C})}{2^{nR}}\\
\leq& \frac{n^{3}R+o(1)}{2^{nR}}\\
=:& 2^{-n(R-\epsilon'')}
\end{align*}
where $\epsilon''$ goes to 0 as $n \to \infty$.
Therefore,
\begin{align*}
H(X^n|\mathcal{C})=&\sum_{x^n \in A_{\epsilon,0}^{(n)}(X)} p(x^n|\mathcal{C})\log \frac{1}{p(x^n|\mathcal{C})}\\
\geq & \sum_{x^n \in A_{\epsilon,0}^{(n)}(X)} p(x^n|\mathcal{C})\log 2^{n(R-\epsilon'')}\\
=&[n(R-\epsilon'')]\sum_{x^n \in A_{\epsilon,0}^{(n)}(X)} p(x^n|\mathcal{C})\\
=&[n(R-\epsilon'')].\end{align*}

Below, We use the Vapnik-Chervonenkis Theorem to show that a regular codebook appears with high probability.

Let $\mathcal{B}=\{\{x^n\}, x^n \in A_{\epsilon,0}^{(n)}(X) \}$. Since
$|\mathcal{B}|=|A_{\epsilon,0}^{(n)}(X)|\leq  2^{n(H_0(X)+\epsilon)}$, for any $A \subseteq \mathcal{X}^{n}$,
$$|\{\{x^n\}\cap A: x^n \in A_{\epsilon,0}^{(n)}(X)\}| \leq 2^{n(H_0(Y)+\epsilon)},$$
and hence VC-d$(\mathcal{B}) \leq n(H_0(X)+\epsilon)$.

Since VC-d$(\mathcal{B})$ is finite for a fixed $n$, we employ the Vapnik-Chervonenkis Theorem under the range space $(\mathcal{X}^{n},\mathcal{B})$. To satisfy (\ref{E:vctheorem2}), let both $\epsilon$ and $\delta$ in (\ref{E:vctheorem1}) be $\frac{\Delta_{\epsilon}nR}{2^{nR}}$, where $\Delta_{\epsilon}:=\max \{ 8\text{VC-d}(\mathcal{B}), 16e  \}.$ Then the Vapnik-Chervonenkis Theorem states that
\begin{align}
\nonumber &\mbox{Pr}\left\{\sup_{x^n \in A_{\epsilon,0}^{(n)}(X) }  \left| \frac{N(x^n|\mathbf{C})}{2^{nR}} - p(x^n) \right| \leq \frac{\Delta_{\epsilon}nR}{2^{nR}}     \right\}\\
\nonumber \geq &1-\frac{\Delta_{\epsilon}nR}{2^{nR}}\\
\to & 1 \ \text{as} \ n \to \infty . \label{E:O1}
\end{align}
Since $\frac{n^3 R}{2^{nR}}  \geq \frac{\Delta_{\epsilon}nR}{2^{nR}} $ for sufficiently large $n$, (\ref{E:O1}) concludes that $\mbox{Pr}(\mathbf{C} \text{~is regular}) \to 1$ as $n \to \infty$.

Therefore,
\begin{align*}
H(X^n|\mathbf{C})=&\sum_{\mathcal{C}} p(\mathcal{C})H(X^n|\mathcal{C})\\
\geq & \sum_{\mathcal{C} \text{~is regular}} p(\mathcal{C})H(X^n|\mathcal{C})\\
\geq &[n(R-\epsilon'')]\sum_{\mathcal{C} \text{~is regular}} p(\mathcal{C})\\
=&[n(R-\epsilon'')](1-o(1)),\end{align*}
and
\begin{align}
\nonumber \lim_{n\to \infty}\frac{1}{n}H(X^n|\mathbf{C})\geq &\lim_{n\to \infty}\frac{1}{n}[n(R-\epsilon'')](1-o(1))\\
\nonumber = & \lim_{n\to \infty} (R-\epsilon'')(1-o(1))\\
=&R.\label{E:cbound2}\end{align}
Combining (\ref{E:cbound1}) and (\ref{E:cbound2}), we finish the proof of Lemma \ref{L:rlessc2}.
\end{proof}

\section{Rate Needed to Compress Relay's Observation}\label{S:Relay}
To study the optimality of the compress-and-forward strategy, in this section, we investigate the rate needed for the relay to losslessly compress its observation. In the classical approach of \cite{covelg79}, the compression scheme at the relay was only based on the distribution used for generating the codebook at the source, without being specific on the codebook generated. However, since both the relay and destination have the knowledge of the exact codebook used at the source, it is natural to ask whether it is beneficial for the relay to compress its observation based on this codebook information. This question motivates us to compare the rates needed to compress the relay's observation in two different scenarios: when the relay uses the knowledge of the source's codebook and when the relay simply ignores this knowledge.

Specifically, we consider the two compression problems shown in Figure \ref{F:compression}, where $Y^n$ is generated from $X^n$ through the channel $(\mathcal{X},p(y|x),\mathcal{Y})$, and $\mathbf{C}$ in (b) is the source's codebook information available to both the encoder and decoder. Interestingly, we will show that to perfectly recover $Y^n$, the minimum required rates in both scenarios are the same when the rate $R$ associated with $\mathbf{C}$ is greater than the channel capacity.

\begin{figure}[hbt]
\centering
\includegraphics[width=3in]{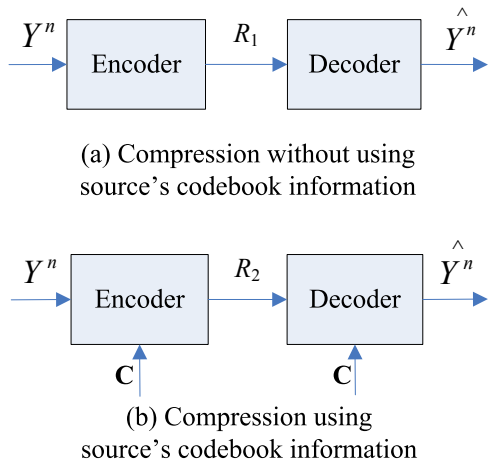}
\caption{Two scenarios where the relay compresses its observation.}
 \label{F:compression}
\end{figure}

Formally, we have the following theorem:

\begin{theorem}\label{T:compression}
For the discrete memoryless channel $(\mathcal{X},p(y|x),\mathcal{Y})$, generate the codebook at random according to $p_0(x)$ and reserve only the $\epsilon$-strongly typical codewords. Let $\mathbf{C}$ be the source's codebook with rate $R$, and $X^n$ and $Y^n$ be the input and output of the channel respectively. When $R>I_{0}(X;Y)$, to compress the channel output $Y^n$, we have

1) $Y^n$ can be encoded at rate $R_1$ and recovered with arbitrarily low probability of error if $R_1 > H_0(Y)$.

2) Given that the source's codebook information $\mathbf{C}$ is available to both the encoder and decoder and $Y^n$ is encoded at rate $R_2$, the decoding probability of error will be bounded away from zero if $R_2 < H_0(Y)$, which implies that we cannot compress the channel output better even if the source's codebook information is employed.
\end{theorem}

To show Theorem \ref{T:compression}, we need the following lemma.

\begin{lemma}\label{L:converse}
For the compression problem in Figure \ref{F:compression}-(b), we can encode $Y^n$ at rate $R_2$ and recover it with the probability of error $P_e^{(n)} \to 0$ only if
\begin{equation}
R_2 \geq \lim _{n \to \infty} \frac{1}{n} H(Y^n|\mathbf{C}).\label{E:converse1}
\end{equation}
\end{lemma}

\begin{proof}[Proof of Lemma \ref{L:converse}]
The source code for Figure \ref{F:compression}-(b) consists of an encoder mapping $f(Y^n, \mathbf{C})$ and a decoder mapping $g(f(Y^n, \mathbf{C}), \mathbf{C})$. Let $I=f(Y^n, \mathbf{C})$, then $P_e^{(n)}=\mbox{Pr}(g(I,\mathbf{C})\neq Y^n)$. By Fano's Inequality, for any source code with $P_e^{(n)} \to 0$, we have
\begin{equation}
H(Y^n|I, \mathbf{C})\leq P_e^{(n)}  \log |\mathcal{Y}^n|+1= P_e^{(n)}n  \log |\mathcal{Y}|+1=n\epsilon_{n},\label{E:fano}
\end{equation}
where $\epsilon_n \to 0$ as $n \to \infty$.

Therefore, for any source code with rate $R_2$ and $P_e^{(n)} \to 0$, we have the following chain of inequalities
\begin{align}
nR_2 &\stackrel {}{\geq} H(I)\label{m}\\
\nonumber &\stackrel {}{\geq} H(I|\mathbf{C}) \\
\nonumber &\stackrel { }{=} H(Y^n,I|\mathbf{C})-H(Y^n|I,\mathbf{C})\\
\nonumber &\stackrel { }{=} H(Y^n|\mathbf{C})+H(I|Y^n, \mathbf{C})-H(Y^n|I,\mathbf{C})\\
&\stackrel {}{=} H(Y^n|\mathbf{C})-H(Y^n|I,\mathbf{C})\label{o}\\
&\stackrel {}{\geq} H(Y^n|\mathbf{C})-n\epsilon_n \label{p}
\end{align}
where (\ref{m}) follows from the fact that $I \in \{1,2,\dots,2^{nR_{2}}\}$, (\ref{o}) follows from the fact that $I$ is a function of $Y^n$ and $\mathbf{C}$, and (\ref{p}) follows from (\ref{E:fano}). Dividing the inequality $nR_2 \geq  H(Y^n|\mathbf{C})-n\epsilon_n$ by $n$ and taking the limit as $n \to \infty$, we establish Lemma \ref{L:converse}.
\end{proof}

\begin{proof}[Proof of Theorem \ref{T:compression}]

Proof of Part 1):
To show Part 1), we only need to show that the sequence $Y^n$ satisfies the Asymptotic Equipartition Property, i.e., $\mbox{Pr}(Y^n \in A_{\epsilon,0}^{(n)}(Y)) \to 1$, as $n \to \infty$. Then, following the classical approach to show the source coding theorem, we can conclude that the rate $R_1>H_0(Y)$ is achievable. By Lemma \ref{L:JAEP}, $\mbox{Pr}((X^n,Y^n) \in A_{\epsilon,0}^{(n)}(X,Y))\to 1 \ \text{as} \ n \to \infty.$
Thus, the sequence $Y^n$ satisfies the Asymptotic Equipartition Property and the rate $R_1>H_0(Y)$ is achievable.

Proof of Part 2): By Lemma \ref{L:converse}, given that the codebook information $\mathbf{C}$ is available to both the encoder and decoder and $Y^n$ is encoded at rate $R_2$, $P_e^{(n)} \to 0$ only if $R_2 \geq \lim _{n \to \infty} \frac{1}{n} H(Y^n|\mathbf{C})$. By Theorem \ref{T:rate}, $\lim _{n \to \infty} \frac{1}{n} H(Y^n|\mathbf{C})= H_0(Y)$ when $R>I_0(X;Y)$. Therefore, when $R>I_0(X;Y)$, $P_e^{(n)} \to 0$ only if $R_2 \geq H_0(Y)$, which establishes Part 2).
\end{proof}

\appendices
\section{Proof of Lemma \ref{L:JAEP}}\label{A:A}

Proof of Part 1): Let $\tilde{X}^n$ be drawn i.i.d. according to $p_0(x)$ and $\tilde{Y}^n$ be generated from $\tilde{X}^n$ through the channel $(\mathcal{X},p(y|x),\mathcal{Y})$. Then, we have
\begin{align*}
& \mbox{Pr}((X^n,Y^n) \in A_{\epsilon,0}^{(n)}(X,Y))\\
=&\sum_{(x^n,y^n) \in A_{\epsilon,0}^{(n)}(X,Y)} p(x^n)p(y^n|x^n)\\
=&  \sum_{(x^n,y^n) \in A_{\epsilon,0}^{(n)}(X,Y)} \mbox{Pr}(\tilde{X}^{n} =x^n|\tilde{X}^n \in A_{\epsilon,0}^{(n)}(X))\cdot \mbox{Pr}(Y^n=y^n|X^n=x^n) \\
=&  \sum_{(x^n,y^n) \in A_{\epsilon,0}^{(n)}(X,Y)} \mbox{Pr}(\tilde{X}^{n} =x^n|\tilde{X}^n \in A_{\epsilon,0}^{(n)}(X))\cdot \mbox{Pr}(\tilde{Y}^n=y^n|\tilde{X}^n=x^n) \\
=&  \sum_{(x^n,y^n) \in A_{\epsilon,0}^{(n)}(X,Y)} \mbox{Pr}((\tilde{X}^n, \tilde{Y}^n)=(x^n,y^n)|\tilde{X}^n \in A_{\epsilon,0}^{(n)}(X))\\
=&  \mbox{Pr}((\tilde{X}^n, \tilde{Y}^n) \in A_{\epsilon,0}^{(n)}(X,Y)| \tilde{X}^n \in A_{\epsilon,0}^{(n)}(X))\\
=& \frac{\mbox{Pr}((\tilde{X}^n, \tilde{Y}^n) \in A_{\epsilon,0}^{(n)}(X,Y))}{\mbox{Pr}(\tilde{X}^n \in A_{\epsilon,0}^{(n)}(X))}\\
\to & 1, \text{as $n \to \infty$.}
\end{align*}

Proof of Part 2): Denote the $\epsilon$-weakly typical sets with respect to $p_0(x)$, $p_0(y)$ and $p_0(x,y)$ by $W_{\epsilon,0}^{(n)}(X)$, $W_{\epsilon,0}^{(n)}(Y)$ and $W_{\epsilon,0}^{(n)}(X,Y)$ respectively. Along the same line as in the proof of part 1), we can prove that $\mbox{Pr}((X^n, Y^n) \in W_{\epsilon,0}^{(n)}(X,Y)) \to 1$ as $n \to \infty$, and hence $\mbox{Pr}(X^n \in W_{\epsilon,0}^{(n)}(X))$ and $\mbox{Pr}(Y^n \in W_{\epsilon,0}^{(n)}(Y))$ both go to 1 as $n \to \infty$.

Now, consider $H(Y^n)$. We have
\begin{align}
\nonumber H(Y^n)=&\sum_{y^n}p(y^n)\log\frac{1}{p(y^n)}\\
\nonumber =&\sum_{y^n \in W_{\epsilon,0}^{(n)}(Y)}p(y^n)\log\frac{1}{p(y^n)}+\sum_{y^n \notin W_{\epsilon,0}^{(n)}(Y)}p(y^n)\log\frac{1}{p(y^n)}\\
\nonumber =&:\phi_1+\phi_2
\end{align}

For $\phi_1$, we have
\begin{align*}
\phi_1=&\sum_{y^n \in W_{\epsilon,0}^{(n)}(Y)}p(y^n)\log\frac{1}{p(y^n)} \\
\nonumber \leq & \sum_{y^n \in W_{\epsilon,0}^{(n)}(Y)}p(y^n)\log 2^{n(H_0(Y)+\epsilon)}\\
\nonumber =&n(H_0(Y)+\epsilon) \mbox{Pr}(Y^n\in W_{\epsilon,0}^{(n)}(Y) )\\
=&n(H_0(Y)+\epsilon) (1-o(1)),
\end{align*}
where the inequality follows from the fact that $p(y^n) \geq 2^{-n(H_0(Y)+\epsilon)}$ for any $y^n \in W_{\epsilon,0}^{(n)}(Y)$.

For $\phi_2$, we have
\begin{align}
\nonumber\phi_2=&\sum_{y^n \notin W_{\epsilon,0}^{(n)}(Y)}p(y^n)\log\frac{1}{p(y^n)} \\
\nonumber=&-\sum_{y^n \in W_{\epsilon,0}^{(n)c}(Y)}p(y^n)\log p(y^n)  \\
\leq & -\left(\sum_{y^n \in W_{\epsilon,0}^{(n)c}(Y)}p(y^n)\right)\log \frac{\sum_{y^n \in W_{\epsilon,0}^{(n)c}(Y)}p(y^n)}{|W_{\epsilon,0}^{(n)c}(Y)|}\label{E:logsum}\\
\nonumber =&-\mbox{Pr}(Y^n\notin W_{\epsilon,0}^{(n)}(Y))\log \frac{\mbox{Pr}(Y^n\notin W_{\epsilon,0}^{(n)}(Y))}{|W_{\epsilon,0}^{(n)c}(Y)|}\\
\nonumber =&-\mbox{Pr}(Y^n\notin W_{\epsilon,0}^{(n)}(Y))\log  \mbox{Pr}(Y^n\notin W_{\epsilon,0}^{(n)}(Y)) + \mbox{Pr}(Y^n\notin W_{\epsilon,0}^{(n)}(Y))\log |W_{\epsilon,0}^{(n)c}(Y)|\\
=& o(1) + \mbox{Pr}(Y^n\notin W_{\epsilon,0}^{(n)}(Y))\log |W_{\epsilon,0}^{(n)c}(Y)| \label{E:fromaep}\\
\nonumber \leq& o(1) + \mbox{Pr}(Y^n\notin W_{\epsilon,0}^{(n)}(Y))\log |\mathcal{Y}|^n\\
\nonumber =& o(1)+ n \cdot \mbox{Pr}(Y^n\notin W_{\epsilon,0}^{(n)}(Y))\log |\mathcal{Y}|\\
=& n \cdot o(1).\label{E:fromaep1}
\end{align}

(\ref{E:logsum}) follows from the the log sum inequality (see Theorem 2.7.1 in \cite{Coverbook}), which states that for non-negative numbers, $a_1,a_2,\ldots,a_n$ and $b_1,b_2,\ldots,b_n$,
\begin{align*}
\sum_{i=1}^{n}a_i \log \frac{a_i}{b_i}\geq \left(\sum_{i=1}^{n}a_i\right) \log \frac{\sum_{i=1}^{n}a_i}{\sum_{i=1}^{n}b_i}
\end{align*}
with equality if and only if $\frac{a_i}{b_i}$ are equal for all $i$.

(\ref{E:fromaep}) and (\ref{E:fromaep1}) both follow from the fact that $\mbox{Pr}(Y^n \in W_{\epsilon,0}^{(n)}(Y)) \to 1$ as $n \to \infty.$

Therefore,
\begin{align}
\nonumber H(Y^n)=&\phi_1+\phi_2\\
\nonumber \leq & n(H_0(Y)+\epsilon) (1-o(1))+n \cdot o(1)\\
=&n(H_0(Y)+\epsilon) (1-o(1)).\label{E:combiningagain1}
\end{align}

Similarly, we have
\begin{align}
\nonumber H(Y^n)\geq &\sum_{y^n \in W_{\epsilon,0}^{(n)}(Y)}p(y^n)\log\frac{1}{p(y^n)}\\
\nonumber \geq & \sum_{y^n \in W_{\epsilon,0}^{(n)}(Y)}p(y^n)\log 2^{n(H_0(Y)-\epsilon)}\\
\nonumber =&n(H_0(Y)-\epsilon) \mbox{Pr}(Y^n\in W_{\epsilon,0}^{(n)}(Y) )\\
=&n(H_0(Y)-\epsilon) (1-o(1)).\label{E:combiningagain2}
\end{align}
Combining (\ref{E:combiningagain1}) and (\ref{E:combiningagain2}), we have $\lim_{n \to \infty}\frac{1}{n}H(Y^n)=H_0(Y)$.

Along the same line as above, we can also prove that $\lim_{n \to \infty}\frac{1}{n}H(X^n)=H_0(X)$ and $\lim_{n \to \infty}\frac{1}{n}H(X^n,Y^n)=H_0(X,Y)$, which concludes the proof of Lemma \ref{L:JAEP}.

\end{document}